\documentclass{llncs}
\usepackage[utf8]{inputenc}
\usepackage[T1]{fontenc}
\usepackage{lmodern}
\usepackage[english]{babel}
\usepackage[babel=true,final]{microtype}
\usepackage{array}
\usepackage{stmaryrd,amssymb,amsmath,mathtools}
\usepackage{braket}
\usepackage{pgfplots}
\usepackage{xfrac}
\usepackage{enumitem}
\usepackage{siunitx}
\usepackage{etoolbox}
\usepackage{bigfoot}

\DeclareNewFootnote[para]{default}

\newtoggle{longversion}
\toggletrue{longversion}
\newcommand{\IfShort}[1]{\nottoggle{longversion}{#1}{}}
\newcommand{\IfLong}[1]{\iftoggle{longversion}{#1}{}}

\newcommand{\TRR}{\mbox{\sf TSRR}}
\newcommand{\TWO}{\mbox{\sf TSTP}}
\newcommand{\TSL}{\mbox{\sf TSMP}}
\newcommand{\SLRthree}{$\mbox{\sf SLR}_3$}
\newcommand{\defiff}{\ensuremath\:\mathbin{:\!\!\iff\!\!}\:}
\DeclareMathOperator{\widen} {\nabla}
\DeclareMathOperator{\narrow}{\Delta}
\DeclareMathOperator{\upd}{\boxslash}
\newcommand{\angl}[1]{{\langle#1\rangle}}
\newcommand{\sem}[1]{\llbracket#1\rrbracket}
\newcommand{\semSh}[1]{\llbracket#1\rrbracket^\sharp}
\newcommand{\R}{\mathrel{\mathcal{R}}}
\newcommand{\N}{\mathbb{N}}
\newcommand{\C}{\mathbb{C}}
\newcommand{\D}{\mathbb{D}}

\renewcommand{\gets}{\coloneqq}
\DeclareMathAlphabet{\kw}{\encodingdefault}{\sfdefault}{bx}{n}

\pgfplotsset{compat=1.5}

\begin{document}
\title{Enforcing Termination of Interprocedural Analysis}
\author{Stefan Schulze Frielinghaus \and Helmut Seidl \and Ralf Vogler}
\institute{	Fakult\"at f\"ur Informatik,
	TU M\"unchen, Germany,\\
	$\{${\tt schulzef,seidl,voglerr}$\}${\tt @in.tum.de}
	}
\maketitle

\begin{abstract}
Interprocedural analysis by means of partial tabulation of summary functions
may not terminate when the same procedure is analyzed for infinitely many
abstract calling contexts or when the abstract domain has infinite strictly ascending chains.
As a remedy, we present a novel local solver for general abstract equation systems, be they 
monotonic or not, and prove that this solver fails to terminate only when infinitely many
variables are encountered.
We clarify in which sense the computed results are sound.
Moreover, we show that interprocedural analysis performed by this novel local solver,
is guaranteed to terminate for all non-recursive programs --- irrespective of whether the 
complete lattice is infinite or has infinite strictly ascending or descending chains.

\end{abstract}

\section{Introduction}\label{s:intro}

It is well known that static analysis of run-time properties of programs
by means of abstract interpretation can be compiled into systems of equations over complete lattices~\cite{Cousot77-1}.
Thereby, various interesting properties require complete
lattices which may have infinite strictly ascending or descending chains~\cite{Gonnord06,Chen10,Bagnara05}.
In order to determine a (post-) solution of a system of equations over such lattices,
Cousot and Cousot propose to perform a first phase of iteration using a \emph{widening} operator to obtain a
post-solution which later may be improved by a second phase of iteration using a \emph{narrowing}
operator.
This strict arrangement into separate phases, though, has the disadvantage that precision unnecessarily
may be given up which later is difficult to recover.
It has been observed that widening and narrowing need not be organized into separate phases~\cite{Apinis13,Apinis16,Amato16}.
Instead various algorithms are proposed which
\emph{intertwine} widening with narrowing in order to compute a (reasonably small) post-fixpoint of the given
system of equations.
The idea there is to combine widening with narrowing into a single operator and then to iterate according to
some fixed ordering over the variables of the system. Still, monotonicity of all right-hand sides is required
for the resulting algorithms to be terminating~\cite{Apinis13,Amato16}.

Non-monotonic right-hand sides, however, are introduced by interprocedural analysis in the style of~\cite{Apinis12}
when partial tabulation of summary functions 
is used.
In order to see this, consider an abstract lattice $\D$ of possible program invariants.
Then the abstract effect of a procedure call can be formalized as a transformation $f^\sharp$ from $\D\to\D$.
For rich lattices $\D$ such transformations may be difficult to represent and compute with.
As a remedy, each single variable function may be decomposed into a set of variables --- one
for each possible argument --- where each such variable now receives values from $\D$ only.
As a result, the difficulty of dealing with elements of $\D\to\D$ is replaced with the difficulty of
dealing with systems
of equations which are infinite when $\D$ is infinite.
Moreover, composition of abstract functions is translated into \emph{indirect addressing} of variables (the outcome of the
analysis for one function call determines for which argument another function is queried) --- implying non-monotonicity~\cite{Fecht96}.
Thus, termination of interprocedural analysis by means of the solvers from~\cite{Apinis13,Amato16} cannot be
guaranteed.
Interestingly, the \emph{local} solver \SLRthree~\cite{Amato16} terminates in many practical cases.
Nontermination, though, may arise in two flavors:
\begin{itemize}
\item	infinitely many variables may be encountered, i.e., some procedure may be analyzed for an ever growing
	number of calling contexts;
\item	the algorithm may for some variable switch infinitely often from a narrowing iteration back to a widening
	iteration.
\end{itemize}
From a conceptual view, the situation still is unsatisfactory: any solver used as a fixpoint engine within
a static analysis tool should reliably terminate under reasonable assumptions.
In this paper, we therefore re-examine interprocedural analysis by means of local solvers.
First, we extend an ordinary local solver to a two-phase solver which performs widening and subsequently
narrowing. The novel point is that both iterations are performed in a demand-driven way so that also during the
narrowing phase fresh variables may be encountered for which no sound over-approximation has yet been computed.

In order to enhance precision of this demand-driven two-phase solver,
we then design a new local solver which intertwines the two phases.
In contrast to the solvers in~\cite{Apinis13,Amato16}, however, we can no longer rely on a fixed combination
of a widening and a narrowing operator, but must enhance the solver with extra logic to decide when to apply which operator.
For both solvers, we prove that
they terminate --- whenever only finitely many variables are encountered:
irrespective whether the abstract system is monotonic or not.
Both solvers are guaranteed to return (partial) post-solutions of the abstract system of equations only
if all right-hand sides are monotonic.
Therefore, we make clear in which sense the computed results are nonetheless sound --- even in the non-monotonic case.
For that, we provide a sufficient condition for an abstract variable assignment to be a sound description
of a concrete system --- given only a (possibly non-monotonic) abstract system of equations.
This sufficient condition is formulated by means of the
\emph{lower monotonization} of the abstract system. Also, we elaborate for partial solutions in which sense the
domain of the returned variable assignment provides sound information. Here, the formalization of purity of functions
based on computation trees and variable dependencies plays a crucial role.
Finally, we prove that interprocedural analysis in the style of~\cite{Cousot77-3,Apinis12} with partial tabulation
using our local solvers terminates for all non-recursive programs and every complete lattice with or without
infinite strictly ascending or descending chains.

The paper is organized as follows.
In Section~\ref{s:basics} we recall the basics of abstract interpretation
and introduce the idea of a lower monotonization of an abstract system of equations.
In Section~\ref{s:narrow} we
recapitulate widening and narrowing. As a warm-up, a terminating variant of round-robin iteration is presented
in Section~\ref{s:rr}.
In Section~\ref{s:local} we formalize the idea of local solvers based on the notion of purity of functions of
right-hand sides of abstract equation systems and provide a theorem indicating in which sense local solvers
for non-monotonic abstract systems compute sound results for concrete systems.
A first local solver is presented in
Section~\ref{s:two} where widening and narrowing is done in conceptually separated phases.
In Section~\ref{s:tsl}, we present a local solver where widening and narrowing is
intertwined.
Section~\ref{s:inter} considers the abstract equations systems encountered by interprocedural
analysis. A concept of stratification is introduced which is satisfied if the programs to be analyzed are
non-recursive. These notions enable us to prove our main result concerning termination of interprocedural analysis
with partial tabulation by means of the solvers from sections~\ref{s:two} and~\ref{s:tsl}.

 \section{Basics on Abstract Interpretation}\label{s:basics}

In the following we recapitulate the basics of abstract interpretation as introduced by Cousot and Cousot~\cite{Cousot77-1,Cousot92}.
Assume that the concrete semantics of a system is described by a system of equations
\begin{equation}
	x = f_x, \quad x\in X	\label{e:concrete}
\end{equation}
where $X$ is a set of variables taking values in some power set
lattice $(\C,\subseteq,\cup)$ where $\C = 2^Q$ for some set $Q$ of concrete program states,
and for each $x\in X$, $f_x: (X\to\C)\to\C$ is the defining right-hand side of $x$.
For the concrete system of equations, we assume that all right-hand sides $f_x, x\in X$, are \emph{monotonic}.
Accordingly, this system of equations has a unique least solution $\sigma$ which can be obtained as the least upper bound
of all assignments $\sigma_\tau$, $\tau$ an ordinal. The assignments
$\sigma_\tau: X\to\C$ are defined as follows.
If $\tau = 0$, then $\sigma_\tau\,x= \bot$ for all $x\in X$.
If $\tau = \tau'+1$ is a successor ordinal, then $\sigma_\tau\,x = f_x\,\sigma_{\tau'}$, and
if $\tau$ is a limit ordinal, then $\sigma_\tau\,x = \bigcup\{f_x\,\sigma_{\tau'}\mid \tau'<\tau\}$.
An \emph{abstract} system of equations
\begin{equation}
	y = f^\sharp_y, \quad y\in Y	\label{e:abstract}
\end{equation}
specifies an analysis of the concrete system of equations. Here, $Y$ is a set of
\emph{abstract} variables which may not necessarily be in one-to-one correspondence to the concrete variables in the set $X$.
The variables in $Y$ take values in some complete lattice $(\D,\sqsubseteq,\sqcup)$ of abstract values
and for every abstract variable $y\in Y$, $f^\sharp_y:(Y \to\D)\to\D$ is the abstract
defining right-hand side of $y$.
The elements $d\in\D$ are meant to represent invariants, i.e., properties of states.
It is for simplicity that we assume the set $\D$ of all possible invariants to form a complete lattice,
as any partial order can be embedded into a complete lattice so that all existing
least upper and greatest lower bounds are preserved~\cite{MacNeille37}.
In order to relate concrete sets of states with abstract values,
we assume that there is a Galois connection between $\C$ and $\D$,
i.e., there are monotonic functions $\alpha:\C\to\D$, $\gamma:\D\to\C$ such that
for all $c\in\C$ and $d\in\D$, $\alpha(c)\sqsubseteq d$ iff $c\subseteq \gamma(d)$.
Between the sets of concrete and abstract variables, we assume that there is a 
\emph{description relation} ${\R} \subseteq X\times Y$.
Via the Galois connection between $\C$ and $\D$, the description relation $\R$
between variables is lifted to a description relation $\R^*$ between assignments
$\sigma:X\to\C$ and $\sigma^\sharp:Y\to\D$ by defining $\sigma\R^*\sigma^\sharp$
iff for all $x\in X,y\in Y$, $\sigma(x)\subseteq\gamma(\sigma^\sharp(y))$ whenever $x \R y$ holds.
Following~\cite{Cousot92}, we do not assume that the right-hand sides of the abstract equation system 
are necessarily monotonic.
For a sound analysis, we only assume that all right-hand sides respect the description relation, i.e.,
that for all $x\in X$ and $y\in Y$ with $x \R y$,
\begin{equation}
	f_x\,\sigma\subseteq\gamma (f_y^\sharp\,\sigma^\sharp)   \label{e:soundrhs}
\end{equation}
whenever $\sigma\R^*\sigma^\sharp$ holds.
 Our key concept for proving soundness of abstract variable assignments w.r.t.\ the concrete system of equations
is the notion of the \emph{lower monotonization} of the abstract system.
For every function $f^\sharp: (Y\to\D)\to\D$ we consider the function
\begin{equation}
\underline f^\sharp\,\sigma = \bigsqcap\Set{ f^\sharp\,\sigma' | \sigma\sqsubseteq\sigma' } \label{e:lower}
\end{equation}
which we call \emph{lower monotonization} of $f^\sharp$.
By definition, we have:
\begin{lemma}\label{l:lower0}
For every function $f^\sharp:(Y\to\D)\to\D$ the following holds:
\begin{enumerate}
\item	$\underline f^\sharp$ is monotonic;
\item	$\underline f^\sharp\,\sigma^\sharp \sqsubseteq f^\sharp\,\sigma^\sharp$ for all $\sigma^\sharp$;
\item	$\underline f^\sharp = f^\sharp$ whenever $f^\sharp$ is monotonic.	\qed
\end{enumerate}
\end{lemma}
The lower monotonization of the abstract system~\eqref{e:abstract} then is defined as the system
\begin{eqnarray}
y &=& \underline f^\sharp_y,\qquad y\in Y \label{e:lmono}
\end{eqnarray}
Since all right-hand sides of~\eqref{e:lmono} are monotonic, this system has a least solution.
\begin{example}\label{e:nonmon}
Consider the single equation
\[\begin{array}{l@{\;\;}c@{\;\;}l}
y_1 &=& {\sf if}\;y_1 = 0\;{\sf then}\; 1\;{\sf else}\;0
\end{array}\]
over the complete lattice of non-negative integers equipped with an infimum element,
i.e., let the domain $\mathbb{D} = \mathbb{N} \cup \{\infty\}$.
This system is not monotonic. Its lower monotonization is given by $y_1 = 0$.
\qed
\end{example}
\begin{lemma}\label{l:lower}\label{l:mono}
Assume that $\sigma$ is the least solution of the concrete system~\eqref{e:concrete}.
Then $\sigma \R^* \sigma^\sharp$ for every post-solution  $\sigma^\sharp$ of the lower monotonization~\eqref{e:lmono}.
\end{lemma}
\begin{proof}
For every ordinal $\tau$, let $\sigma_\tau$ denote the $\tau$th approximation of
the least solution of the concrete system and assume that
$\sigma^\sharp$ is a post-solution of the lower monotonization of the abstract system, i.e., 
$
\sigma^\sharp\,y \sqsupseteq \underline f^\sharp_y\,\sigma^\sharp
$ holds for all $y\in Y$.
By ordinal induction, we prove that $\sigma_\tau \R^* \sigma^\sharp$.
The claim clearly holds for $\tau=0$.
First assume that $\tau=\tau'+1$ is a successor ordinal, and that the claim holds for $\tau'$, i.e., $\sigma_{\tau'} \R^* \sigma^\sharp$.
Accordingly, $\sigma_{\tau'} \R^* \sigma'$ holds for all $\sigma' \sqsupseteq\sigma^\sharp$.
Consider any pair of variables $x,y$ with $x\R y$. Then
$
\sigma_\tau\,x = f_x\sigma_{\tau'}\subseteq\gamma(f_y^\sharp\,\sigma')
$
for all $\sigma' \sqsupseteq \sigma^\sharp$. Accordingly,
$
\alpha(\sigma_\tau\,x) \sqsubseteq f_y^\sharp\,\sigma'
$
for all $\sigma' \sqsupseteq \sigma^\sharp$, and therefore,
\[
\alpha(\sigma_\tau\,x)
        \sqsubseteq\bigsqcap\Set{ f_y^\sharp\,\sigma' | \sigma' \sqsupseteq \sigma^\sharp }
	= \underline f^\sharp_y\,\sigma^\sharp
	\sqsubseteq \sigma^\sharp\,y
\]
since $\sigma^\sharp$ is a post-solution.
From that, the claim follows for the ordinal $\tau$.
Now assume that $\tau$ is a limit ordinal, and that the claim holds for all ordinals $\tau'<\tau$.
Again consider any pair of variables $x,y$ with $x\R y$. Then
\[
\sigma_\tau\,x = \bigcup\Set{ \sigma_{\tau'}\,x | \tau' <\tau }
		\subseteq \bigcup\Set{\gamma(\sigma^\sharp\,y) | \tau' <\tau }
		= \gamma(\sigma^\sharp\,y )
\]
and the claim also follows for the limit ordinal $\tau$.
\qed
\end{proof}
From Lemma~\ref{l:lower} we conclude that for the abstract system from Example~\ref{e:nonmon}
the assignment $\sigma^\sharp = \{y_1 \mapsto 0\}$ is a sound description of every corresponding concrete system,
since $\sigma^\sharp$ is a post-solution of the lower monotonization $y_1 = 0$.

In general, Lemma~\ref{l:mono} provides us with a sufficient condition guaranteeing that
an abstract assignment $\sigma^\sharp$ is
sound w.r.t.\ the concrete system~\eqref{e:concrete} and the description relation $\R$,
namely, that $\sigma^\sharp$ is a post-solution of the system~\eqref{e:lmono}.
This sufficient condition is remarkable as it is an \emph{intrinsic} property of the abstract system
since it does not refer to the concrete system.
As a corollary we obtain:
\begin{corollary}\label{c:post}
Every post-solution $\sigma^\sharp$ of the abstract system~\eqref{e:abstract} is sound.
\end{corollary}
\begin{proof}
For all $y\in Y$, 
$
\sigma^\sharp\,y\sqsupseteq f^\sharp_y\,\sigma^\sharp \sqsupseteq \underline f^\sharp_y\,\sigma^\sharp
$ holds.
Accordingly, $\sigma^\sharp$ is a post-solution of the lower monotonization of the abstract system
and therefore sound.
\qed
\end{proof}

 \section{Widening and Narrowing}\label{s:narrow}

It is instructive to recall the basic algorithmic approach to determine non-trivial post-solutions
of abstract systems \eqref{e:abstract} when the set $Y$ of variables is finite, 
all right-hand sides are monotonic and the complete lattice $\D$ has finite strictly increasing chains only.
In this case, 
\emph{chaotic iteration}
may be applied. This kind of iteration starts with the initial assignment $\underline\bot$
which assigns $\bot$ to every variable $y\in Y$ and then repeatedly evaluates right-hand sides to update the
values of variables until the values for all variables have stabilized.
This method may also be applied if right-hand sides are non-monotonic: the only modification required is
to update the value for each variable not just with the new value provided by the left-hand side, but with
some upper bound of the old value for a variable with the new value.
As a result, a \emph{post-solution} of the system is computed which, according to Corollary~\ref{c:post},
is sound.

The situation is more intricate, if the complete lattice in question has strictly ascending chains of infinite length.
Here, we follow Cousot and Cousot \cite{Cousot77-1,Cousot92,Cousot15} who suggest to accelerate iteration
by means of \emph{widening} and \emph{narrowing}.
A widening operator $\widen:\D\times\D\to\D$ takes the old value $a\in\D$ and a new value $b\in\D$
and combines them to a value $a\sqcup b\sqsubseteq a\widen b$ with the additional understanding that
for any sequence $b_i,i\geq 0,$ and any value $a_0$, the sequence $a_{i+1} = a_i\widen b_i,i\geq 0$, is
ultimately stable.
In contrast, a narrowing operator $\narrow:\D\times\D\to\D$ takes the old value $a\in\D$ and a new value 
$b\in\D$ and combines them to a value $a\narrow b$ satisfying
$a\sqcap b\sqsubseteq a\narrow b\sqsubseteq a$ --- with the additional understanding 
that for any sequence $b_i,i\geq 0,$ and any value $a_0$, the sequence $a_{i+1} = a_i\narrow b_i,i\geq 0$, is
ultimately stable.

While the widening operator is meant to reach a post-solution after a finite number of updates
to each variable of the abstract system,
the narrowing operator allows to improve upon a variable assignment once it is known to be sound.
In particular, if all right-hand sides are monotonic, the result of a narrowing iteration, if started with
a post-solution of the abstract system, again results in a post-solution. Accordingly, the returned 
variable assignment can easily be verified to be sound.
In analyzers which iterate according to the syntactical structure of programs such as {\sc Astree} \cite{Astree},
this strict separation into two phases, though, has been given up. There, when iterating over one loop, narrowing for the current
loop is triggered as soon as locally a post-solution has been attained.
This kind of intertwining widening and narrowing is systematically explored in \cite{Apinis13,Amato16}. 
There, a widening operator is combined with a narrowing operator into a single derived operator $\upd$ defined by
\[
a\upd b = \begin{array}[t]{l}
	{\bf if}\; b\sqsubseteq a\;{\bf then}\; a\narrow b	\\
	{\bf else}\; a\widen b
	\end{array}
\]
also called \emph{warrowing}.
Solvers which perform chaotic iteration and use warrowing to combine old values with new contributions,
necessarily return post-solutions --- whenever they terminate. In \cite{Apinis13,Apinis16}, termination
could only be guaranteed for systems of equations where all right-hand sides are monotonic. 
For \emph{non-monotonic} systems as may occur at interprocedural analysis, only practical evidence could 
be provided for the proposed algorithms to terminate in interesting cases.

Here, our goal is to lift these limitations by providing solvers which terminate for all finite
abstract systems of equations and all complete lattices --- no matter whether right-hand sides are monotonic 
or not. For that purpose, we dissolve the operator $\upd$ again into its components. 
Instead, we equip the solving routines with extra logic to decide when to apply which operator.

 \section{Terminating Structured Round-Robin Iteration}\label{s:rr}

Let us consider a finite abstract system as given by:
\begin{equation}
y_i	= f^\sharp_i,\qquad i=1,\dotsc,n
		\label{e:abstract-finite}
\end{equation}
In~\cite{Apinis13}, a variation of round-robin iteration is presented which 
is guaranteed to terminate for monotonic systems, while it may not terminate
for non-monotonic systems. In order to remedy this failure, we re-design this algorithm
by additionally maintaining a flag which indicates whether the variable presently under consideration
has or has not reached a sound value (Fig.~\ref{f:rr}).
\begin{figure}
\[
\begin{array}{l}
{\bf void}\;{\sf solve}(b,i)\;\{	\\
\qquad{\bf if}\; (i \leq 0)\;{\bf return};	\\
\qquad{\sf solve}(b,i-1);	\\
\qquad{\it tmp} \gets f^\sharp_i\,\sigma;	\\
\qquad{b'} \gets b;	\\
\qquad{\bf if}\; (b)\;{\it tmp} \gets \sigma[y_i] \narrow {\it tmp};	\\
\qquad{\bf else}\; {\bf if}\; ({\it tmp} \sqsubseteq \sigma[y_i])\;\{	\\
\qquad\qquad{\it tmp} \gets \sigma[y_i] \narrow {\it tmp};	\\
\qquad\qquad{b'} \gets {\bf true};	\\
\qquad\}\;{\bf else}\;{\it tmp} \gets \sigma[y_i] \widen{\it tmp};	\\
\qquad{\bf if}\; (\sigma[y_i] = {\it tmp})\;{\bf then}\;{\bf return};	\\
\qquad\sigma[y_i] \gets {\it tmp};	\\
\qquad{\sf solve}({b'},i);	\\
\}	\\
\end{array}
\]
\caption{\label{f:rr}Terminating structured round-robin iteration.}
\end{figure}
Solving starts with a call ${\sf solve}({\bf false},n)$ where $n$ is the highest priority of a
variable.
A variable $y_i$ has a higher priority than a variable $y_j$ whenever $i > j$ holds.
A call ${\sf solve}(b,i)$ considers variables up to priority $i$ only. The Boolean argument $b$ indicates whether a
sound abstraction (relative to the current values of the higher priority variables) has already been reached.
The algorithm first iterates on the lower priority variables (if there are any).
Once solving of these is completed, the right-hand side $f^\sharp_i$ of the current variable $y_i$ is evaluated
and stored in the variable ${\it tmp}$.
Additionally, $b'$ is initialized with the Boolean argument $b$.
First assume that $b$ has already the value ${\bf true}$.
Then the old value $\sigma\,y_i$ is combined with the new value in ${\it tmp}$ by means of
the narrowing operator giving the new value of ${\it tmp}$. If that is equal to the old value,
we are done and ${\sf solve}$ returns.
Otherwise, $\sigma\,y_i$ is updated to ${\it tmp}$, and ${\sf solve}({\bf true},i)$ is called tail-recursively.
Next assume that $b$ has still the value ${\bf false}$.
Then the algorithm distinguishes two cases. If the
old value $\sigma\,y_i$ exceeds the new value, the variable ${\it tmp}$ receives the combination of
both values by means of the narrowing operator. Additionally, $b'$ is set to ${\bf true}$.
Otherwise, the new value for ${\it tmp}$ is obtained by means of widening.
Again, if the resulting value of ${\it tmp}$ is equal to the current value $\sigma\,y_i$ of $y_i$,
the algorithm returns, whereas if they differ, then
$\sigma\,y_i$ is updated to ${\it tmp}$ and the algorithm recursively calls itself for the
actual parameters $(b',n)$.
In light of Theorem~\ref{t:trr}, the
resulting algorithm is called \emph{terminating structured round-robin iteration}
or \TRR\ for short.

\begin{theorem}\label{t:ssrr}\label{t:trr}
The algorithm in Figure~\ref{f:rr}
terminates for all finite abstract systems of the form~\eqref{e:abstract-finite}.
Upon termination, it returns a variable assignment $\sigma$ which is sound.
If all right-hand sides are monotonic, $\sigma$ is a post-solution.
\end{theorem}
\IfShort{A proof is provided in the long version of this paper~\cite{longversion}.}\IfLong{For a proof see Appendix~\ref{s:ssrr-term:proof}.}In fact, for monotonic systems, 
the new variation 
of round-robin iteration behaves identical to
the algorithm {\sf SRR} from~\cite{Apinis13}.

 \section{Local Solvers}\label{s:local}

Local solving may gain efficiency by querying the value only of a hopefully small subset of variables whose values
still are sufficient to answer the initial query.
Such solvers are at the heart of program analysis frameworks such as the {\sc Ciao} system~\cite{Hermenegildo05,Hermenegildo12} or
{\sc Goblint}.
In order to reason about \emph{partial} variable assignments as computed by local solvers,
we can no longer consider right-hand sides in equations as black boxes, but require a notion of
\emph{variable dependence}.

For the concrete system we assume that right-hand sides are mathematical functions
of type $(X\to\C)\to\C$ where for any such function $f$ and variable assignment
$\sigma:X\to\C$, we are given a superset ${\sf dep}(f,\sigma)$ of variables onto which $f$
possibly depends, i.e.,
\begin{eqnarray}\label{e:dep}
(\forall x\in {\sf dep}(f,\sigma).\,\sigma[x]=\sigma'[x]) &\Longrightarrow &	f\,\sigma=f\,\sigma'
\end{eqnarray}
for all $\sigma':X\to\C$.
Let $\sigma:X\to\C$ denote a solution of the concrete system. Then we call a subset $X'\subseteq X$
of variables $\sigma$-\emph{closed}, if for all $x\in X'$, ${\sf dep}(f_x,\sigma)\subseteq X'$. Then
for every $x$ contained in the $\sigma$-closed subset $X'$, $\sigma[x]$ can be determined
already if the values of $\sigma$ are known for the variables in $X'$ only.
 
In~\cite{mcctr-fixpt,Fecht00} it is observed that for suitable formulations of interprocedural analysis, the
set of all run-time calling contexts of procedures can be extracted from $\sigma$-closed sets of variables.
\begin{example}\label{e:inter}
The following system may arise from the analysis of a program consisting of a procedure with
a loop (signified by the program point $u$) within which the same procedure
is called twice in a row. Likewise, the procedure $p$ iterates on some program point $v$
by repeatedly applying the function $g$:
\[
\begin{array}{lll}
\angl{u,q} &=& \bigcup\Set{\angl{v,q_1} | q_1\in \bigcup\Set{\angl{v,q_2} | q_2\in\angl{u,q}}}\cup \{q\}        \\
\angl{v,q} &=& \bigcup\Set{g\,q_1 | q_1\in\angl{v,q}}\cup \{q\}
\end{array}
\]
for $q\in Q$. 
Here, $Q$ is a superset of all possible system states, and the unary function $g:Q\to 2^Q$ describes 
the operational behavior of the body of the loop at $v$.
The set of variables of this system is given by $X=\{\angl{u,q},\angl{v,q}\mid q\in Q\}$
where $\angl{u,q},\angl{v,q}$ represent the sets of program states possibly occurring at program points $u$ and $v$,
respectively, when the corresponding procedures have been called in context $q$.
For any variable assignment $\sigma$, the dependence sets of the right-hand sides are naturally defined by:
\[
\begin{array}{lll}
{\sf dep}(f_{\angl{u,q}},\sigma) &=& \{\angl{u,q}\}\cup\{\angl{v,q_2} \mid q_2\in \sigma\,\angl{u,q}\}	\\
&&\phantom{\{\angl{u,q}\}}
		\cup\{\angl{v,q_1}\mid q_2\in\sigma\angl{u,q}, q_1\in\sigma\,\angl{v,q_2}\}     \\
{\sf dep}(f_{\angl{v,q}},\sigma) &=& \{\angl{v,q}\}
\end{array}
\]
where $f_x$ again denotes the right-hand side function for a variable $x$.
Assuming that $g(q_0) = \{q_1\}$ and $g(q_1) = \emptyset$, the least solution $\sigma$ maps
$\angl{u,q_0},\angl{v,q_0}$ to the set $\{q_0,q_1\}$ and $\angl{u,q_1},\angl{v,q_1}$ to $\{q_1\}$.
Accordingly, 
the set $\{\angl{u,q_i},\angl{v,q_i}\mid i=0,1\}$ is $\sigma$-closed.
We conclude, given the program is called with initial context $q_0$, that the procedure $p$
is called with contexts $q_0$ and $q_1$ only.
\qed
\end{example}

\noindent
In concrete systems of equations, right-hand sides may depend on \emph{infinitely} many variables.
Since abstract systems are meant to give rise to effective algorithms, 
we impose more restrictive assumptions onto their right-hand side functions.
For these, we insist that only finitely many variables may be queried.
Following the considerations in~\cite{Hofmann10a,Hofmann10b,Karbyshev13Thesis}, we demand that every
right-hand side $f^\sharp$ of the abstract system is \emph{pure} in the sense of~\cite{Hofmann10a}.
This means that, operationally, the evaluation of $f^\sharp$ for any abstract variable assignment $\sigma^\sharp$
consists of a finite sequence of variable look-ups before eventually, a value is returned.
Technically, $f^\sharp$ can be represented by a \emph{computation tree}, i.e., 
is an element of
\[
{\sf tree} \Coloneqq {\sf Answer}\;\D\quad\mid\quad{\sf Query}\;Y\times(\D\to{\sf tree})
\]
Thus, a computation tree either is a leaf immediately containing a value or a query, which consists of
a variable together with a continuation which, for every possible value of the variable returns a tree
representing the remaining computation.
Each computation tree defines a function $\sem{t}:(Y\to\D)\to\D$ by:
\[
\begin{array}{lll}
\sem{{\sf Answer}\,d}\;\sigma	&=& d	\\
\sem{{\sf Query}\,(y,c)}\;\sigma	&=& \sem{c\,(\sigma[y])}\,\sigma
\end{array}
\]
Following~\cite{Hofmann10a}, the tree representation is uniquely determined by (the operational semantics of) $f^\sharp$.

\begin{example}\label{e:comp_tree}
Computation trees can be considered as generalizations of binary decision diagrams to arbitrary sets $\D$.
For example, let $\D = \mathbb{N} \cup \{\infty\}$, i.e., the natural numbers (equipped with the natural ordering
and extended with $\infty$ as top element), the function $f^\sharp:(Y\to\D)\to\D$
with $\{y_1,y_2\}\subseteq Y$, defined by
\[
f^\sharp\,\sigma = 
{\bf if}\;\sigma[y_1] > 5\;{\bf then}\;1 +\sigma[y_2]\;{\bf else}\;\sigma[y_1]
\]
is represented by the tree
\begin{flalign*}
&&
\begin{array}[b]{lll}
{\sf Query}\,(y_1,{\bf fun}\,d_1 &\to& {\bf if}\;d_1 > 5\;{\bf then}\;{\sf Query}\,(y_2,{\bf fun}\, d_2\to {\sf Answer}\,(1+d_2)) \\
&&{\bf else}\;{\sf Query}\,(y_1,{\bf fun}\,d_1\to {\sf Answer}\,d_1))
\end{array}
&&
\null\qed
\end{flalign*}
\end{example}
A set ${\sf dep}(f^\sharp,\sigma^\sharp)\subseteq Y$ with a property analogous to~\eqref{e:dep} can be explicitly
obtained from the tree representation $t$ of $f^\sharp$
by defining
${\sf dep}(f^\sharp,\sigma^\sharp)={\sf treedep}(t,\sigma^\sharp)$ where:
\[
\begin{array}{lll}
{\sf treedep}({\sf Answer}\,d,\sigma^\sharp)	&=&	\emptyset	\\
{\sf treedep}({\sf Query}\,(y,c),\sigma^\sharp)	&=&
			\{y\}\cup{\sf treedep}(c\,(\sigma^\sharp[y]),\sigma^\sharp)
\end{array}
\]
Technically, this means that the value $f^\sharp\,\sigma^\sharp = \sem{t}\,\sigma^\sharp$ can be
computed already for \emph{partial} variable assignments $\sigma':Y'\to\D$, whenever
${\sf dep}(f^\sharp,\underline\top\oplus\sigma')= {\sf treedep}(t,\underline\top\oplus\sigma')\subseteq Y'$.
Here, $\underline\top:Y\to\D$ maps each variable of $Y$ to $\top$
and $\underline\top\oplus\sigma'$ returns the value $\sigma'[y]$ for every $y\in Y'$ and
$\top$ otherwise.
\begin{example}\label{e:comp_tree_1}
Consider the function $f^\sharp$ from Example~\ref{e:comp_tree} together with the partial
assignment $\sigma'= \{y_1\mapsto 3\}$. Then ${\sf dep}(f^\sharp,\underline\top\oplus\sigma') = \{y_1\}$. 
\qed
\end{example}

\noindent
We call a partial variable assignment $\sigma':Y'\to\D$ \emph{closed} (w.r.t.\ an
abstract system~\eqref{e:abstract}),
if for all $y\in Y'$, ${\sf dep}(f_y^\sharp,\underline\top\oplus\sigma')\subseteq Y'$.

In the following, we strengthen the description relation $\R$ 
additionally to take variable dependencies into account.
We say that
the abstract system~\eqref{e:abstract} \emph{simulates} the concrete system~\eqref{e:concrete}
(relative to the description relation $\R$) iff for all pairs $x,y$ of variables with $x\R y$, such that
for the concrete and abstract right-hand sides $f_x$ and $f^\sharp_y$, respectively,
property~\eqref{e:soundrhs} holds and additionally 
${\sf dep}(f_x,\sigma)\R{\sf dep}(f^\sharp_y,\sigma^\sharp)$ whenever $\sigma \R^* \sigma^\sharp$.
Here, a pair of sets $X',Y'$ of concrete and abstract variables is in relation $\R$ if
for all $x\in X'$, $x\R y$ for some $y\in Y'$.
Theorem~\ref{p:local} demonstrates the significance of closed abstract assignments which are sound.

\begin{theorem}\label{p:local}
Assume that the abstract system~\eqref{e:abstract} simulates the concrete system~\eqref{e:concrete}
(relative to $\R$)
where $\sigma$ is the least solution of the concrete system.
Assume that $\sigma^\sharp:Y'\to\D$ is a partial assignment with the following properties:
\begin{enumerate}
\item	$\sigma^\sharp$ is closed;
\item	$\underline\top\oplus\sigma^\sharp$ is a post-solution of the lower monotonization of the
abstract system.
\end{enumerate}
Then the set $X'=\Set{x\in X | \exists\,y\in Y'.\, x\R y}$ is $\sigma$-closed.
\end{theorem}

\begin{proof}
By Lemma~\ref{l:mono}, $\sigma \R^* (\underline\top\oplus\sigma^\sharp)$ holds.
Now assume that $x\R y$ for some $y\in Y'$. By definition therefore,
${\sf dep}(f_x,\sigma)\R{\sf dep}(f^\sharp_y,\underline\top\oplus\sigma^\sharp)$.
Since the latter is a subset of $Y'$,
the former must be a subset of $X'$, and the assertion follows.
\qed
\end{proof}

 \section{Terminating Structured Two-Phase Solving}\label{s:two}

We first present a local version of a two-phase algorithm to determine a
sound variable assignment for an abstract system of equations. As the algorithm is local, 
no pre-processing of the equation system is possible. Accordingly, variables where 
widening or narrowing is to be applied must be determined dynamically
(in contrast to solvers based on static variable dependencies where widening points can be statically determined~\cite{Bourdoncle1993}).
We solve this problem by assigning \emph{priorities} to variables in decreasing order
in which they are encountered, and consider a variable as a candidate for widening/narrowing
whenever it is queried during the evaluation of a lower priority variable.
The second issue is that during the narrowing iteration of the second phase, 
variables may be encountered which have not yet been seen and for which therefore no sound 
approximation is available.
In order to deal with this situation, the algorithm does not maintain a single variable assignment, 
but two distinct ones. While assignment $\sigma_0$ is used for the widening phase, $\sigma_1$ is used
for narrowing with the understanding that, once the widening phase is completed, the value of a variable $y$
from $\sigma_0$ is copied as the initial value of $y$ into $\sigma_1$. This clear distinction allows to
continue the widening iteration for every newly encountered variable $y'$ in order to determine an acceptable
initial value before continuing with the narrowing iteration.
The resulting algorithm can be found in Figures~\ref{f:two} and~\ref{f:twoprime}.
\begin{figure}
\[
\begin{array}[t]{l}
\qquad{\bf void}\;{\sf iterate}_0(n)\;\{	\\
\qquad\qquad{\bf if}\;(Q\neq\emptyset \land {\sf min\_prio}(Q) \leq n)\;\{	\\
\qquad\qquad\qquad y \gets {\sf extract\_min}(Q);	\\
\qquad\qquad\qquad\phantom{{\sf solve}_1(y,n);}	\\
\qquad\qquad\qquad{\sf do\_var}_0(y);	\\
\qquad\qquad\qquad{\sf iterate}_0(n);	\\
\qquad\qquad\}	\\
\qquad\}	\\
\qquad{\bf void}\;{\sf solve}_0(y)\;\{	\\
\qquad\qquad{\bf if}\;(y\in{\sf dom}_0)\;{\bf return};	\\
\qquad\qquad{\sf dom}_0 \gets {\sf dom}_0\cup\{y\};	\\
\qquad\qquad{\sf prio}[y] \gets {\sf next\_prio}();	\\
\qquad\qquad\sigma_0[y]\gets \bot;\\
\qquad\qquad{\sf infl}[y]\gets\emptyset;	\\
\qquad\qquad{\sf do\_var}_0(y);	\\
\qquad\qquad{\sf iterate}_0({\sf prio}[y]);	\\
\qquad\}	\\
\end{array}
\begin{array}[t]{l}
\qquad{\bf void}\;{\sf iterate}_1(n)\;\{	\\
\qquad\qquad{\bf if}\;(Q\neq\emptyset \land {\sf min\_prio}(Q) \leq n)\;\{	\\
\qquad\qquad\qquad y \gets {\sf extract\_min}(Q);	\\
\qquad\qquad\qquad{\sf solve}_1(y,{\sf prio}[y]-1);	\\
\qquad\qquad\qquad{\sf do\_var}_1(y);	\\
\qquad\qquad\qquad{\sf iterate}_1(n);	\\
\qquad\qquad\}	\\
\qquad\}	\\
\qquad{\bf void}\;{\sf solve}_1(y,n)\;\{	\\
\qquad\qquad{\bf if}\;(y\in {\sf dom}_1)\;{\bf return};	\\
\qquad\qquad{\sf solve}_0(y);\\
\qquad\qquad{\sf dom}_1 \gets {\sf dom}_1\cup\{y\};	\\
\qquad\qquad\sigma_1[y]\gets\sigma_0[y];	\\
\qquad\qquad{\bf forall}\;(z\in \{y\}\cup{\sf infl}[y])\;{\sf insert}\;z\;Q;\\
\qquad\qquad{\sf infl}[y] \gets \emptyset;	\\
\qquad\qquad{\sf iterate}_1(n);			\\
\qquad\}		\\
\end{array}
\]
\caption{\label{f:two}The solver \TWO, part 1.}
\end{figure}
\begin{figure}
\[
\begin{array}[t]{l}
\qquad{\bf void}\;{\sf do\_var}_0(y)\;\{        \\
\qquad\qquad{\it isp} \gets y \in {\sf point}; \\
\qquad\qquad{\sf point} \gets {\sf point} \backslash \{y\};\\
\qquad\qquad{\mathbb D}\;{\sf eval}_0(z)\;\{      \\
\qquad\qquad\qquad{\sf solve}_0(z);     \\
\qquad\qquad\qquad{\bf if}\;({\sf prio}[z] \geq {\sf prio}[y])\;{\sf point} \gets {\sf point}\cup\{z\}; \\
\qquad\qquad\qquad{\sf infl}[z] \gets {\sf infl}[z]\cup\{y\};   \\
\qquad\qquad\qquad{\bf return}\;\sigma_0[z]); \\
\qquad\qquad\}  \\
\qquad\qquad{\it tmp} \gets f^\sharp_y\;{\sf eval}_0;     \\
\qquad\qquad{\bf if}\;({\it isp})\;{\it tmp} \gets \sigma_0[y]\widen{\it tmp};        \\
\qquad\qquad{\bf if}\;(\sigma_0[y] = {\it tmp})\;{\bf return};       \\
\qquad\qquad\sigma_0[y]\gets{\it tmp};\\
\qquad\qquad{\bf forall}\;(z\in {\sf infl}[y])\;{\sf insert}\;z\;Q;\\
\qquad\qquad{\sf infl}[y] \gets \emptyset;      \\
\qquad\qquad{\bf return};       \\
\qquad\}
\end{array}
\begin{array}[t]{l}
\qquad{\bf void}\;{\sf do\_var}_1(y)\;\{	\\
\qquad\qquad{\it isp} \gets y \in {\sf point}; \\
\qquad\qquad{\sf point} \gets {\sf point} \backslash \{y\};\\
\qquad\qquad{\mathbb D}\;{\sf eval}_1(z)\;\{	\\
\qquad\qquad\qquad{\sf solve}_1(z,{\sf prio}[y]-1);      \\
\qquad\qquad\qquad{\bf if}\;({\sf prio}[z] \geq {\sf prio}[y])\;{\sf point} \gets {\sf point}\cup\{z\};	\\
\qquad\qquad\qquad{\sf infl}[z] \gets {\sf infl}[z]\cup\{y\};	\\
\qquad\qquad\qquad{\bf return}\;\sigma_1[z];	\\
\qquad\qquad\}	\\
\qquad\qquad{\it tmp} \gets f^\sharp_y\;{\sf eval}_1; 	\\
\qquad\qquad{\bf if}\;({\it isp})\;{\it tmp} \gets \sigma[y]\narrow{\it tmp};	\\
\qquad\qquad{\bf if}\;(\sigma_1[y] = {\it tmp})\;{\bf return};	\\
\qquad\qquad\sigma_1[y]\gets{\it tmp};\\
\qquad\qquad{\bf forall}\;(z\in {\sf infl}[y])\;{\sf insert}\;z\;Q;\\
\qquad\qquad{\sf infl}[y] \gets \emptyset;	\\
\qquad\qquad{\bf return};	\\
\qquad\}
\end{array}
\]
\caption{\label{f:twoprime}The solver \TWO, part 2.}
\end{figure}
Initially, the priority queue $Q$ and the sets ${\sf dom}_0$ and ${\sf dom}_1$ are empty. Accordingly, the mappings
$\sigma_i:{\sf dom}_i\to\D$ and ${\sf infl}:{\sf dom}\to 2^Y$ are also empty.
Likewise, the set ${\sf point}$ is initially empty.
Solving for the variable $y_0$ starts with the call ${\sf solve}_1(y_0,0)$.

Let us first consider the functions ${\sf solve}_0,{\sf iterate}_0,{\sf do\_var}_0$.
These are meant to realize a local widening iteration.
A call ${\sf solve}_0(y)$ first checks whether $y\in{\sf dom}_0$. If this is the case, solving immediately terminates.
Otherwise, $\sigma_0[y]$ is initialized with $\bot$, $y$ is added to ${\sf dom}_0$, 
the empty set is added to ${\sf infl}[y]$, and $y$ receives the next available
priority by means of the call ${\sf next\_prio}$.
Subsequently, ${\sf do\_var}_0(y)$ is called, followed by a call to ${\sf iterate}_0({\sf prio}[y])$
to complete the widening phase for $y$.
Upon termination, a call ${\sf iterate}_0(n)$ for an integer $n$ has removed all variables
of priority at most $n$ from the queue $Q$.
It proceeds as follows. If $Q$ is empty or contains only variables of priority exceeding $n$, it immediately returns.
Otherwise, the variable $y$ with least priority is extracted from $Q$. Having processed ${\sf do\_var}_0(y)$,
the iteration continues with the tail-recursive call ${\sf iterate}_0(n)$.

It remains to describe the function ${\sf do\_var}_0$.  When called for a variable $y$,
the algorithm first determines whether or not $y$ is a widening/narrowing point, i.e., contained in the set
${\sf point}$. If so, $y$ is removed from ${\sf point}$, 
and the flag ${\it isp}$ is set to ${\bf true}$.
Otherwise, ${\it isp}$ is set to ${\bf false}$.
Then the right-hand side $f^\sharp_y$ is evaluated and the result stored in the variable ${\it tmp}$.
For its evaluation, the function $f^\sharp_y$, however, does not receive the current variable assignment $\sigma_0$
but an auxiliary function ${\sf eval}_0$ which serves as a wrapper to the assignment $\sigma_0$.
The wrapper function ${\sf eval}_0$, when queried for a variable $z$, first calls ${\sf solve}_0\,(z)$
to compute a first non-trivial value for $z$.
If the priority of $z$ is greater or equal to the priority of $y$, a potential widening point
is detected. Therefore, $z$ is added to the set ${\sf point}$.
Subsequently, the fact that $z$ was queried during the evaluation of the right-hand side of $y$,
is recorded by adding $y$ to the set ${\sf infl}[z]$.
Finally, 
$\sigma_0[z]$ is returned.

Having evaluated $f^\sharp_y\,{\sf eval}_0$ and stored the result in ${\it tmp}$,
the function ${\sf do\_var}_0$ then applies widening only if ${\it isp}$ equals ${\bf true}$.
In this case, ${\it tmp}$ receives the value of $\sigma[y]\widen{\it tmp}$.
In the next step, ${\it tmp}$ is compared with the current value $\sigma_0[y]$.
If both values are equal, the procedure returns.
Otherwise, $\sigma_0[y]$ is updated to ${\it tmp}$.
The variables in ${\sf infl}[y]$ are inserted into the queue $Q$, and the set ${\sf infl}[y]$
is reset to the empty set.
Only then the procedure returns.

The functions ${\sf solve}_1$, ${\sf iterate}_1$ and ${\sf do\_var}_1$, on the other hand,
are meant to realize the narrowing phase. They essentially work analogously to the corresponding functions
${\sf solve}_0$, ${\sf iterate}_0$ and ${\sf do\_var}_0$. In particular, they re-use the mapping
${\sf infl}$ which records the currently encountered variable dependencies as well as the variable
priorities and the priority queue $Q$.
Instead of $\sigma_0$, ${\sf dom}_0$, however, they now refer to $\sigma_1$, ${\sf dom}_1$, respectively.
Moreover, there are the following differences.

First, the function ${\sf solve}_1$ now receives not only a variable,
but a pair of an integer $n$ and a variable $y$. When called, the function 
first checks whether $y\in{\sf dom}_1$. If this is the case, solving immediately terminates.
Otherwise, ${\sf solve}_0(y)$ is called first. After that call, the widening phase for $y$ 
is assumed to have terminated where the resulting value is $\sigma_0[y]$. 
Accordingly, $\sigma_1[y]$ is initialized with $\sigma_0[y]$, and $y$ is added to ${\sf dom}_1$.
As the value of $\sigma_1$ for $y$ has been updated, 
$y$ together with all variables in ${\sf infl}[y]$ are added to the queue,
whereupon ${\sf infl}[y]$ is set to the empty set, and ${\sf iterate}_1(n)$ is called
to complete the narrowing phase up to the priority $n$.
Upon termination, a call ${\sf iterate}_1(n)$ for an integer $n$ has removed all variables
of priority at most $n$ from the queue $Q$. 
In distinction to ${\sf iterate}_0$, however, it may extract variables $y$ from $Q$ which have not yet been
encountered in the present phase of iteration, i.e., are not yet included
in ${\sf dom}_1$ and thus have not yet received a value in $\sigma_1$. 
To ensure initialization, 
${\sf solve}_1(y,n)$ is called for $n={\sf prio}[y]-1$. This choice of the extra parameter $n$ ensures that
all lower priority variables have been removed from $Q$ before ${\sf do\_var}_1(y)$ is
called.

It remains to explain the function ${\sf do\_var}_1(y)$. 
Again, it essentially behaves like ${\sf do\_var}_0(y)$ --- with the distinction that the narrowing operator
is applied instead of the widening operator. Furthermore, the auxiliary local function ${\sf eval}_0$
is replaced with ${\sf eval}_1$ which now uses a call to ${\sf solve}_1$ for the initialization of its
argument variable $z$ (instead of ${\sf solve}_0$) where the extra integer argument is given by ${\sf prio}[y]-1$,
i.e., an iteration is performed to remove all variables from $Q$ with priorities lower than the 
priority of $y$ (not of $z$).

In light of Theorem~\ref{t:local-two}, we call 
the algorithm from Figures~\ref{f:two} and~\ref{f:twoprime}
\emph{terminating structured two-phase solver}.
\begin{theorem}\label{t:local-two}\label{t:term-two}
The local solver \TWO\ from Figure~\ref{f:two} and~\ref{f:twoprime}
when started with a call ${\sf solve}_1(y_0,0)$ for a variable $y_0$,
terminates for every system of equations whenever only finitely
many variables are encountered.

Upon termination, assignments $\sigma_i^\sharp:Y_i\to\D$, $i=0,1$ are obtained for finite sets $Y_0\supseteq Y_1$ 
of variables so that the following holds:
\begin{enumerate}
\item	$y_0\in Y_1$;
\item	$\sigma_0^\sharp$ is a closed partial post-solution of the abstract system~\eqref{e:abstract};
\item	$\sigma_1^\sharp$ is a closed partial assignment such that
	$\underline\top\oplus\sigma_1^\sharp$ is a post-solution of the lower monotonization of the 
	abstract system~\eqref{e:abstract}.
\end{enumerate}
\end{theorem}
\IfShort{A proof is provided in the long version of this paper~\cite{longversion}.}\IfLong{For a proof see Appendix~\ref{s:two-term:proof}.}

 \section{Terminating Structured Mixed-Phase Solving}\label{s:tsl}

The draw-back of the two-phase solver \TWO\ from the last section is that it may lose precision
already in very simple situations. 

\begin{example}\label{e:compare}
Consider the system:
\[
\begin{array}{lll@{\qquad}lll@{\qquad}lll}
         y_1 &=&	{\sf max}(y_1, y_2)	&
       	 y_2 &=&	{\sf min}(y_3, 2)	&
	 y_3 &=&	y_2+1
\end{array}
\]
over the complete lattice ${\mathbb N}^\infty$ and the
following widening and narrowing operators:
\[
\begin{array}{lll}
a\widen b	&=& {\bf if}\,a< b\,{\bf then}\;\infty\;{\bf else}\; a	\\
a\narrow b	&=& {\bf if}\,a=\infty\,{\bf then}\;b\;{\bf else}\;a
\end{array}
\]
Then ${\sf solve}_0(y_1)$ detects $y_2$ as the only widening point resulting in
\[
\sigma_0 = \{ y_1\mapsto \infty, y_2\mapsto\infty, y_3\mapsto\infty\}
\]
A call to ${\sf solve}_1(y_1,0)$ therefore initializes $y_1$ with $\infty$ implying that
$\sigma_1[y_1] =\infty$ irrespective of the fact that $\sigma_1[y_2] = 2$.
\qed
\end{example}
We may therefore aim at intertwining the two phases into one --- without sacrificing the termination guarantee. 
The idea is to operate on a single variable assignment only and iterate on each variable first in widening and
then in narrowing mode. In order to keep soundness, after every update of a variable $y$ in the widening phase,
all possibly influenced lower priority variables are iterated upon until all stabilize with widening and
narrowing. Only then the widening iteration on $y$ continues. If on the other hand an update for $y$ occurs 
during narrowing, the iteration on possibly influenced lower priority variables is with narrowing only.
The distinction between the two modes of the iteration is maintained by a flag where
${\bf false}$ and ${\bf true}$ correspond to the widening and narrowing phases, respectively.
The algorithm is provided in Figure~\ref{f:slrone} and~\ref{f:slrtwo}.
\begin{figure}
\[
\begin{array}[t]{l}
\qquad{\bf void}\;{\sf iterate}(b,n)\;\{	\\
\qquad\qquad{\bf if}\;(Q\neq\emptyset \land {\sf min\_prio}(Q) \leq n)\;\{	\\
\qquad\qquad\qquad y \gets {\sf extract\_min}(Q);	\\
\qquad\qquad \qquad{b}' \gets {\sf do\_var}(b,y);	\\
\qquad\qquad\qquad n' \gets {\sf prio}[y];	\\
\qquad\qquad\qquad{\bf if}\;(b\neq{b}'\wedge n>n')\;\{	\\
\qquad\qquad\qquad\qquad{\sf iterate}(b',n');	\\
\qquad\qquad\qquad\qquad{\sf iterate}(b,n);	\\
\qquad\qquad\qquad\}\;{\bf else}\;{\sf iterate}(b',n);	\\
\qquad\qquad\}	\\
\qquad\}	\\
\end{array}
\qquad\begin{array}[t]{l}
\qquad{\bf void}\;{\sf solve}(y)\;\{	\\
\qquad\qquad{\bf if}\;(y\in {\sf dom})\;{\bf return};	\\
\qquad\qquad{\sf dom} \gets {\sf dom}\cup\{y\};	\\
\qquad\qquad{\sf prio}[y] \gets {\sf next\_prio}();	\\
\qquad\qquad\sigma[y]\gets \bot;\\
\qquad\qquad{\sf infl}[y]\gets\emptyset;	\\
\qquad\qquad{b}' \gets {\sf do\_var}({\bf false},y);	\\
\qquad\qquad{\sf iterate}({b}',{\sf prio}[y]);	\\
\qquad\}	\\
\end{array}
\]
\caption{\label{f:slrone}The solver \TSL, part 1.}
\end{figure}
\begin{figure}
\[
\begin{array}[t]{l}
\qquad{\bf bool}\;{\sf do\_var}(b,y)\;\{	\\
\qquad\qquad{\it isp} \gets y \in {\sf point}; \\
\qquad\qquad{\sf point} \gets {\sf point} \backslash \{y\};\\
\qquad\qquad{\mathbb D}\;{\sf eval}(z)\;\{	\\
\qquad\qquad\qquad{\sf solve}(z);	\\
\qquad\qquad\qquad{\bf if}\;({\sf prio}[z] \geq {\sf prio}[y])\;{\sf point} \gets {\sf point}\cup\{z\};	\\
\qquad\qquad\qquad{\sf infl}[z] \gets {\sf infl}[z]\cup\{y\};	\\
\qquad\qquad\qquad{\bf return}\;\sigma[z];	\\
\qquad\qquad\}	\\
\qquad\qquad{\it tmp} \gets f^\sharp_y\;{\sf eval}; 	\\
\qquad\qquad{b}' \gets b;	\\
\qquad\qquad\ldots
\end{array}
\begin{array}[t]{l}
\qquad\qquad\ldots	\\
\qquad\qquad{\bf if}\;({\it isp})\\
\qquad\qquad\qquad{\bf if}\;(b)\;{\it tmp} \gets \sigma[y]\narrow{\it tmp};	\\
\qquad\qquad\qquad{\bf else}\;{\bf if}\;({\it tmp} \sqsubseteq \sigma[y])\;\{	\\
\qquad\qquad\qquad\qquad{\it tmp} \gets \sigma[y]\narrow {\it tmp};	\\
\qquad\qquad\qquad\qquad{b}' \gets {\bf true};	\\
\qquad\qquad\qquad\;\}\;{\bf else}\;{\it tmp} \gets \sigma[y]\widen{\it tmp};	\\
\qquad\qquad{\bf if}\;(\sigma[y] = {\it tmp})\;{\bf return}\;{\bf true};	\\
\qquad\qquad\sigma[y] \gets {\it tmp};	\\
\qquad\qquad{\bf forall}\;(z\in {\sf infl}[y])\;{\sf insert}\;z\;Q;\\
\qquad\qquad{\sf infl}[y] \gets \emptyset;	\\
\qquad\qquad{\bf return}\;{b}';	\\
\qquad\}
\end{array}
\]
\caption{\label{f:slrtwo}The solver \TSL, part 2.}
\end{figure}

\noindent
Initially, the priority queue $Q$ and the set ${\sf dom}$ are empty. Accordingly, the mappings
$\sigma:{\sf dom}\to\D$ and ${\sf infl}:{\sf dom}\to Y$
are also empty.
Likewise, the set ${\sf point}$ is initially empty.
Solving for the variable $y_0$ starts with the call ${\sf solve}(y_0)$.
Solving for some variable $y$ first checks whether $y\in{\sf dom}$. If this is the case,
solving immediately terminates.
Otherwise, $y$ is added to ${\sf dom}$ and receives the next available priority
by means of a call to ${\sf next\_prio}$.
That call should provide a value which is less than any priority of a variable in ${\sf dom}$.
Subsequently, the entries $\sigma[y]$ and ${\sf infl}[y]$ are
initialized to $\bot$ and the empty set, respectively, and ${\sf do\_var}$ is called for the pair $({\bf false},y)$.
The return value of this call is stored in the Boolean variable $b'$.
During its execution, this call may have inserted further variables into the queue $Q$.
These are dealt with by the call ${\sf iterate}(b',{\sf prio}[y])$.

Upon termination, a call ${\sf iterate}(b,n)$ has removed all variables of priority at most $n$ from the queue $Q$.
It proceeds as follows. If $Q$ is empty or contains only variables of priority exceeding $n$, it immediately returns.
Otherwise, the variable $y$ with least priority $n'$ is extracted from $Q$. For $(b,y)$, ${\sf do\_var}$ is called
and the return value of this call is stored in $b'$.

Now we distinguish several cases.
If $b={\bf true}$, then the value $b'$ returned by ${\sf do\_var}$
will necessarily be ${\bf true}$ as well. In that case, iteration proceeds by tail-recursively calling
again ${\sf iterate}({\bf true},n)$.
If on the other hand $b={\bf false}$, then the value $b'$ returned by ${\sf do\_var}$
can be either ${\bf true}$ or ${\bf false}$. If $b'={\bf false}$ or $b'={\bf true}$ and $n'=n$, then
${\sf iterate}(b',n)$  is tail-recursively called. If, however, $b'={\bf true}$ and $n> n'$,
then first a sub-iteration is triggered for $({\bf true},n')$ before the main loop proceeds with the call
${\sf iterate}({\bf false},n)$.

It remains to describe the function ${\sf do\_var}$.
When called for a pair $(b,y)$ consisting of a Boolean value $b$ and variable $y$,
the algorithm first determines whether or not $y$ is a widening/narrowing point, i.e., contained in the set
${\sf point}$. If so, $y$ is removed from ${\sf point}$, and the flag ${\it isp}$ is set to ${\bf true}$.
Otherwise, ${\it isp}$ is just set to ${\bf false}$.
Then the right-hand side $f^\sharp_y$ is evaluated and the result stored in the variable ${\it tmp}$.
For its evaluation, the function $f^\sharp_y$, however, does not receive the current variable assignment $\sigma$,
but an auxiliary function ${\sf eval}$ which serves as a wrapper to $\sigma$.
The wrapper function ${\sf eval}$, when queried for a variable $z$, first calls ${\sf solve}\,z$
to compute a first non-trivial value for $z$.
If the priority of $z$ exceeds or is equal to the priority of $y$, a potential widening/narrowing point
is detected. Therefore, $z$ is added to the set ${\sf point}$.
Subsequently, the fact that the value of $z$ was queried during the evaluation of the right-hand side of $y$,
is recorded by adding $y$ to the set ${\sf infl}[z]$.
Finally, the value $\sigma[z]$ is returned.

Having evaluated $f^\sharp_y\,{\sf eval}$ and stored the result in ${\it tmp}$,
the function ${\sf do\_var}$ then decides whether to apply widening or narrowing or none of them
according to the following scheme.
If ${\it isp}$ has not been set to ${\bf true}$, no widening or narrowing is applied.
In this case, the flag $b'$ receives the value $b$.
Therefore now consider the case ${\it isp}={\bf true}$.
Again, the algorithm distinguishes three cases.
If $b={\bf true}$, then necessarily narrowing is applied, i.e., ${\it tmp}$ is updated to
the value of $\sigma[y]\narrow{\it tmp}$, and $b'$ still equals $b$, i.e., ${\bf true}$.
If $b={\bf false}$ then narrowing is applied whenever ${\it tmp}\sqsubseteq\sigma[y]$ holds.
In that case, ${\it tmp}$ is set to $\sigma[y]\narrow{\it tmp}$, and $b'$ to ${\bf true}$.
Otherwise, i.e., if $b={\bf false}$ and ${\it tmp}\not\sqsubseteq\sigma\,y$, then
widening is applied by setting ${\it tmp}$ to $\sigma[y]\widen{\it tmp}$, and $b'$ obtains the value
${\bf false}$.

In the next step, ${\it tmp}$ is compared with the current value $\sigma[y]$.
If both values are equal, the value of $b'$ is returned.
Otherwise, $\sigma[y]$ is updated to ${\it tmp}$.
The variables in ${\sf infl}[y]$ are inserted into the queue $Q$, and the set ${\sf infl}[y]$
is reset to the empty set.
Only then the value of $b'$ is returned.

\begin{example}\label{e_compare1}
Consider the system of equations from Example~\ref{e:compare}.
Calling ${\sf solve}$ for variable $y_1$ will assign the priorities $0,{-}1,{-}2$ to the variables
$y_1,y_2$ and $y_3$, respectively. Evaluation of the right-hand side of $y_1$ proceeds only after 
${\sf solve}(y_2)$ has terminated. During the first update of $y_2$, $y_2$ is inserted into 
the set ${\sf point}$, implying that at the subsequent evaluation the widening operator is applied resulting
in the value $\infty$ for $y_2$ and $y_3$. The subsequent narrowing iteration on $y_2$ and $y_3$ 
improves these values to $2$ and $3$, respectively. Only then the value for $y_1$ is determined which is $2$.
During that evaluation, $y_1$ has also been added to the set ${\sf point}$. The repeated evaluation of its right-hand
side, will however, again produce the value $2$ implying that the iteration terminates with the assignment
\begin{flalign*}
 && \sigma = \{ y_1\mapsto 2, y_2\mapsto 2, y_3\mapsto 3\} &&\null\qed
\end{flalign*}
\end{example}
In light of Theorem~\ref{t:tslr}, we call the algorithm from Figures~\ref{f:slrone} and~\ref{f:slrtwo},
\emph{terminating structured mixed-phase solver} or \TSL\ for short.
\begin{theorem}\label{t:local}\label{t:sslr-term}\label{t:tslr}
The local solver \TSL\ from Figure~\ref{f:slrone} and~\ref{f:slrtwo}
when started for a variable $y_0$,
terminates for every system of equations whenever only finitely
many variables are encountered.

Upon termination, an assignment $\sigma^\sharp:Y_0\to\D$ is returned where $Y_0$ is the set of variables encountered
during ${\sf solve}({\bf false },y_0)$ such that the following holds:
\begin{itemize}
\item	$y_0\in Y_0$,
\item	$\sigma^\sharp$ is a closed partial assignment such that
	$\underline\top\oplus\sigma^\sharp$ is a post-solution of the lower monotonization of the
	abstract system~\eqref{e:abstract}.
\end{itemize}
\end{theorem}
\IfShort{A proof is provided in the long version of this paper~\cite{longversion}.}\IfLong{For a proof see Appendix~\ref{s:slr-term:proof}.}

 \section{Interprocedural Analysis}\label{s:inter}

As seen in example \ref{e:inter}, the concrete semantics of programs with
procedures can be formalized by a system of equations over a set of variables
$X = \{ \angl{u,q} \mid u \in U, q \in Q \}$ where
$U$ is a finite set of program points and $Q$ is the set of possible system states.
A corresponding abstract system of equations for interprocedural analysis can be formalized
using abstract variables from the set
$Y = \{ \angl{u,a} \mid u \in U, a \in \D \}$ where the complete lattice $\D$ of abstract values may
also serve as the set of abstract calling contexts for which each program point $u$ may be analyzed.
The description relation $\R$ between concrete and abstract variables is then given by
$\angl{u,q} \R \angl{u,a} \defiff q \in \gamma(a)$
for all $\angl{u,q} \in X$ and $\angl{u,a} \in Y$ and program points $u \in U$.
Moreoever, we require that for all right-hand sides $f_x$ of the concrete system and
$f^\sharp_y$ of the abstract system that $f_x\,q \subseteq \gamma(f^\sharp_y\,a)$ holds,
whenever $x \R y$ and $q \in \gamma(a)$.
Right-hand sides for abstract variables are given by expressions $e$ according to
the following grammar:
\[
e\quad\Coloneqq\quad	d
		\mid	\alpha
		\mid	g^\sharp\,e_1\cdots\,e_k
		\mid	\angl{u,e}
\]
where $d\in\D$ denotes arbitrary constants,
$\alpha$ is a dedicated variable representing the current calling context,
$g^\sharp:\D\to\cdots\to\D$ is a $k$-ary
function, and $\angl{u,e}$ with $u\in U$ refers to a variable of the equation system.
Each expression $e$ describes a function $\semSh{e}:\D\to (Y\to\D)\to\D$ which is defined by:
\[
\begin{array}[t]{lll}
\semSh{d}\,a\,\sigma	&=&	d	\\
\semSh{\alpha}\,a\,\sigma	&=&	a	\\
\end{array}\qquad\qquad
\begin{array}[t]{lll}
\semSh{g^\sharp\,e_1\cdots e_k}\,a\,\sigma	&=& g^\sharp\,(\semSh{e_1}\,a\,\sigma) \cdots (\semSh{e_k}\,a\,\sigma)	\\
\semSh{\angl{u,e}}\,a\,\sigma	&=& \sigma\,\angl{u,\semSh{e}\,a\,\sigma}
\end{array}
\]
A finite representation of the abstract system of equations then is given by the finite set of schematic equations
\[
\angl{u,\alpha} = e_u,\qquad u\in U
\]
for expressions $e_u$. Each schematic equation $\angl{u,\alpha} =e_u$ denotes
the (possibly infinite) family of equations for the variables $\angl{u,a},a\in\D$.
For each $a\in\D$, the right-hand side function of $\angl{u,a}$ is given by
the function $\semSh{e_u}\,a$.
This function is indeed pure for every expression $e_u$ and every $a\in\D$.
Such systems of equations have been used, e.g., in~\cite{Cousot77-2,Apinis12} to specify interprocedural analyses.
\begin{example}\label{e:inter-1}
Consider the schematic system:
\[
\begin{array}{lll@{\qquad}lll}
\angl{u,\alpha} &=& \angl{v,\angl{v,\angl{u,\alpha}}}\sqcup\alpha    &
\angl{v,\alpha} &=& g^\sharp\,\angl{v,\alpha}\sqcup \alpha
\end{array}
\]
for some unary function $g^\sharp:\D\to\D$.
The resulting abstract system simulates the concrete system from Example~\ref{e:inter},
if $g(q) \subseteq \gamma(g^\sharp(a))$ holds whenever $q\in\gamma(a)$.
\qed
\end{example}
As we have seen in the example, function calls result in indirect addressing via nesting of variables.
In case that the program does not have recursive procedures,
there is a mapping $\lambda:U\to\N$ so that for every $u$ with current calling context $\alpha$, right-hand side $e_u$ and every
subexpression $\angl{u',e'}$ of $e_u$ the following holds:
\begin{itemize}
\item	If $\lambda(u')=\lambda(u)$, then $e'=\alpha$;
\item	If $\lambda(u')\neq\lambda(u)$, then $\lambda(u') < \lambda(u)$.
\end{itemize}
If this property is satisfied, we call the equation scheme \emph{stratified} where $\lambda(u)$ is
the \emph{level} of $u$.
Intuitively, stratification means that a new context is created only for some point $u'$ of a strictly lower level.
For the interprocedural analysis as formalized, e.g., in \cite{Apinis12}, all program points of a given procedure may 
receive the same level while the level decreases whenever another procedure is called.
The system from Example~\ref{e:inter-1} is stratified: we may, e.g., define
$\lambda(u)=2$ and $\lambda(v)=1$.
\begin{theorem}\label{t:inter-term}
The solver \TWO\ as well as the solver \TSL\ terminate for stratified equation schemes.
\end{theorem}
\begin{proof}
We only consider the statement of the theorem for solver \TSL{}.
Assume we run the solver \TSL\ on an abstract system specified by a stratified equation scheme.
In light of Theorem~\ref{t:sslr-term}, it suffices to prove that for every $u\in U$, only finitely
many contexts $a\in\D$ are encountered during fixpoint computation.
First, we note that variables $\angl{v,a}$ may not be influenced by variables $\angl{u,a'}$ with
$\lambda(u)>\lambda(v)$.
Second, let $D_{v,a}$ denote the set of variables $\angl{u,a'}$ with $\lambda(u) =\lambda(v)$
onto which $\angl{v,a'}$ may depend. Then all these variables share the same context.
We conclude that new contexts for a point $v$ at some level $k$ are created only by
the evaluation of right-hand sides of variables of smaller levels.
For each level $k$, let $U_k\subseteq U$ denote the set of all $u$ with $\lambda(u)\leq k$.
We proceed by induction on $k$.
Assume that we have proven termination for all calls ${\sf solve}\,\angl{u,a'}$, $\lambda(u)<k$
for any subset of variables $\angl{u',a''}$ which have already been solved.
Then evaluating a call ${\sf solve}\,\angl{v,a}$ with $\lambda(v)=k$ will either query the
values of other variables $\angl{v',a'}$ where $\lambda(v')=k$. In this case, $a'=a$.
Therefore, only finitely many of these are encountered.
Or variables $\angl{v',a'}$ are queried with $\lambda(v')<k$. For those which have not yet been encountered
${\sf solve}\,\angl{v',a'}$ is called. By induction hypothesis, all these calls terminate
and therefore query only finitely many variables.
As the evaluation of ${\sf call}\,\angl{v,a}$ encounters only finitely many variables, it
terminates.
\qed
\end{proof}
A similar argument explains why interprocedural analyzers based on the functional approach of Sharir/Pnueli
\cite{SharirPnueli81,AltMartin95} terminate not only for finite domains but also for full constant propagation --- 
if only the programs are non-recursive.

 \section{Conclusion}\label{s:conclusion}

We have presented local solvers which are guaranteed to terminate for all abstract systems of equations
given that only finitely many variables are encountered --- irrespective of whether right-hand sides of
the equations are monotonic or not or whether the complete lattice has infinite strictly ascending/descending 
chains or not. 
Furthermore, we showed that interprocedural analysis with partial tabulation
of procedure summaries based on these solvers is guaranteed to terminate with the only assumption
that the program has no recursive procedures.
Clearly, theoretical termination proofs may only give an indication that the proposed algorithms
are well-suited as fixpoint engines within a practical analysis tool.
Termination within reasonable time and space bounds is another issue. 
The numbers provided by our preliminary practical experiments within the analysis framework 
{\sc Goblint} seem encouraging (see Appendix~\ref{s:experiments}).
Interestingly, a direct comparison of the two-phase versus mixed-phase solver for full context-sensitive interprocedural
analysis, indicated that \TSL\ was virtually always faster, while the picture w.r.t.\ precision is not so clear.
Also, the new solvers always returned post-solutions of the abstract systems --- although 
they are not bound to do so.

There are several ways how this work can be extended. Our techniques crucially require a Galois connection to relate the concrete with the abstract domain. It is not clear how this restriction can be lifted. Also one may think of
extending two phased approaches to a many-phase iteration as suggested in~\cite{Cousot15}.

\appendix

\bibliographystyle{abbrv}

\section{Experimental Evaluation}\label{s:experiments}

We implemented the solvers \TWO{} and \TSL{} presented in Sections~\ref{s:two} and~\ref{s:tsl} within the
analysis framework {\sc Goblint}\footnote{http://goblint.in.tum.de/}.
For that, these solvers have been extended to deal with \emph{side-effects}
(see~\cite{Apinis12} for a detailed discussion of this mechanism) to jointly deal with
flow- and context-sensitive and flow-insensitive analyses.
In order to perform a fair comparison of the new solvers with \emph{warrowing}-based local solving
as proposed in~\cite{Apinis13,Amato16}, we provided a simplified version of \TSL{}. This simplified
solver performs priority based iteration in the same way as \TSL{} but uses the warrowing operator
instead of selecting operators according to extra flags.
These three solvers were evaluated on the SPECint benchmark suite\footnote{https://www.spec.org/cpu2006/CINT2006/}
consisting of not too small real-world C programs (\num{1600} to \num{34000} LOC). Furthermore, the following
C programs where analyzed:
\mbox{ent}\footnote{http://www.fourmilab.ch/random/ (version 28.01.2008)},
\mbox{figlet}\footnote{http://www.figlet.org/},
\mbox{maradns}\footnote{http://www.maradns.org/},
\mbox{wget}\footnote{https://www.gnu.org/s/wget/},
and some programs from the \mbox{coreutils}\footnote{https://www.gnu.org/s/coreutils/} package.
The analyzed program \mbox{wget} is the largest one with around \num{77000} LOC\@.

The analyses which we performed are put on top of a basic analysis of pointers, strings and enums.
For {\sf enum} variables, \emph{sets} of possibles values are maintained.
The benchmark programs were analyzed with full context-sensitivity of local data while globals were
treated flow-insensitively.

The experimental setting is a fully context-sensitive interval analysis of {\sf int} variables.
Therefore, program \mbox{482.sphinx} had to be excluded from the benchmark suite since it uses
procedures which recurse on {\sf int} arguments.
Interestingly, the \emph{warrowing} solver behaves exactly the same as \TSL{} on all of our benchmark programs.
We interpret this by the fact that for the given analysis the right-hand sides are effectively monotonic.
Accordingly, Figure~\ref{f:bench2} only reports the relative precision of \TWO\ compared to \TSL\@.

\begin{figure}[tb]
\begin{center}
\begin{tikzpicture}[scale=.77]
\pgfplotstableread{test		vars			evals              name
1	        51.269123783		69.5024924895      401.bzip2
2		52.3623445826		75.1937984496      429.mcf
3		50.6051368175		59.7864590675      456.hmmer
4		50.1007387508		58.4844559585      458.sjeng
5		50.0000000000		67.5421403371      ent
6		50.8718664096		73.130971067       figlet-2.2.5
7		49.9690519708		59.4152879835      maradns-1.4.06
8		50.6558914265		62.2711225929      wget-1.12
9               48.8167475728	        62.3060344828      coreutils-8.13-cksum
10              40.6539637732	        47.204669003       coreutils-8.13-cp
11              50.7379627389	        67.7289554161      coreutils-8.13-cut
12              50.7870675841	        63.7449802529      coreutils-8.13-dd
13              50.7023885787	        66.5717192269      coreutils-8.13-df
14              51.4159598407	        71.9275323553      coreutils-8.13-du
15              46.3858892352	        53.8288954927      coreutils-8.13-mv
16              50.7917656374	        68.2744855772      coreutils-8.13-nohup
17              52.9022797322	        85.2016708588      coreutils-8.13-pr
18              51.2813922749	        68.1925493662      coreutils-8.13-ptx
19              50.5519059749	        63.0456402332      coreutils-8.13-tail
}\tableA
\begin{axis}[
  height=7cm,
  width=15cm,
  legend style={at={(0.5,0.89)},anchor=south,legend columns=-1},
  legend style={draw=none},
  ybar=0pt,
  bar width=7pt,
  xticklabels from table={\tableA}{name},
  xtick=data,
  x tick label style={rotate=45,anchor=east},
  xtick pos=left,
  yticklabel={\pgfmathprintnumber\tick\%},
]
\addplot table[x=test,y=vars] {\tableA};
\addplot table[x=test,y=evals] {\tableA};
\legend{\#Vars \sfrac{\TSL}{\TWO},\#Evals \sfrac{\TSL}{\TWO}}
\end{axis}
\end{tikzpicture}
\end{center}
The blue bars (resp.\ red bars) depict the percentage of required variables
(resp.\ evaluations of right-hand sides) of the solver \TSL{} compared to the solver \TWO{}.
\caption{\label{f:bench1}Efficiency of \TSL{} vs.\ \TWO{}.}
\end{figure}

Figure~\ref{f:bench1} compares the solvers \TSL{} and
\TWO{} in terms of space and time. For a reasonable metric for space we choose the
total number of variables (i.e., occurring pairs of program points and contexts) and for time
the total number of evaluations of right-hand sides of a corresponding variable.
The table indicates that the solver \TSL{} requires only around half of the variables of the solver \TWO{}.
Interestingly, the percentage of evaluations of right-hand sides in a run of \TSL{} is still around 60 to 70
percent of the solver \TWO{}.

\begin{figure}[tb]
\begin{center}
\begin{tikzpicture}[scale=.77]
\pgfplotstableread{test	eq	slr	tp	uk        name
1	92.4	5.2	2.3	0.1       401.bzip2
2	86.4	9	4.6	0         429.mcf
3	97.4	1.5	1.1	0         456.hmmer
4	94.1	0	5.9	0         458.sjeng
5	81.3	0	18.7	0         ent
6	79.6	18.8	1.6	0         figlet-2.2.5
7	87.6	1.8	10.5	0.1       maradns-1.4.06
8       94	2.3	3.5	0.2       wget-1.12
9       99.1	0.0	0.9	0.0       coreutils-8.13-cksum
10      90.6	4.4	5.0	0.0       coreutils-8.13-cp
11      94.7	4.5	0.8	0.0       coreutils-8.13-cut
12      96.1	3.0	0.9	0.0       coreutils-8.13-dd
13      93.3	4.3	2.4	0.0       coreutils-8.13-df
14      94.3	4.4	1.3	0.0       coreutils-8.13-du
15      96.0	2.5	1.5	0.0       coreutils-8.13-mv
16      94.3	5.4	0.3	0.0       coreutils-8.13-nohup
17      81.3	13.8	1.9	2.9       coreutils-8.13-pr
18      93.3	6.1	0.6	0.0       coreutils-8.13-ptx
19      91.5	3.4	5.1	0.0       coreutils-8.13-tail
}\tableB
\begin{axis}[
  enlarge y limits=0.15,
  height=7cm,
  width=15cm,
  legend style={at={(0.5,0.89)},anchor=south,legend columns=-1},
  legend style={draw=none},
  ybar stacked,
  xticklabels from table={\tableB}{name},
  xtick=data,
  x tick label style={rotate=45,anchor=east},
  xtick pos=left,
  xtick align=outside,
  yticklabel={\pgfmathprintnumber\tick\%},
]
\addplot table[x=test,y=eq] {\tableB};
\addplot table[x=test,y=slr] {\tableB};
\addplot table[x=test,y=tp] {\tableB};
\addplot table[x=test,y=uk] {\tableB};
\legend{(a),(b),(c),(d)}
\end{axis}
\end{tikzpicture}
\end{center}
Percentage of variables occurring during a a run of both solvers
\begin{enumerate}[label=(\alph*)]
  \item for which \TSL{} and \TWO{} compute the same value;
  \item for which \TSL{} computes more precise results then \TWO{};
  \item for which \TWO{} computes more precise results then \TSL{};
  \item for which the results computed by \TSL{} and \TWO{} are incomparable.
\end{enumerate}
\caption{\label{f:bench2}Precision of \TSL{} vs. \TWO{}.}
\end{figure}
In the second experiment, as depicted by Figure~\ref{f:bench2}, we compare the precision of the two solvers.
For a reasonable metric we only compare the values of variables which occur in both solver runs.
As a result, between 80 and 95 percent of the variables receive the same value.
It is remarkable that the mixed phase solver is not necessarily more precise.
In many cases, an increase in precision is observed, yes --- there are, however, also input programs
where the two-phase solver excels.
 
\IfLong{  \section{Proof of Theorem~\ref{t:ssrr}}\label{s:ssrr-term:proof}

By induction on $i$, we prove that ${\sf solve}(b,i)$ terminates.
For $i=0$, the statement is obviously true. Now assume that $i>0$ and that
by induction hypothesis, ${\sf solve}(b',i-1)$ terminates for $b'\in\{{\bf true},{\bf false}\}$.
First consider the case where $b={\bf true}$. In this case, the flag $b'$ for the tail-recursive call will be 
equal to ${\bf true}$ as well,
and only narrowing will be applied to $y_i$. Therefore, the sequence of tail-recursive calls
${\sf solve}(b,i)$ eventually will terminate.
Now consider the case where $b={\bf false}$. By induction hypothesis, all recursive calls ${\sf solve}(b',i-1)$
terminate. Consider the sequence of tail recursive calls where the flag $b'$ is not set to ${\bf true}$.
Within this sequence, the new values for $y_i$ form an ascending chain $d_0\sqsubset d_1 \sqsubset d_2\ldots$
where $d_{j+1} = d_j\widen a_j$ for suitable values $a_j$. Due to the properties of a widening operator,
this sequence is finite, i.e., there is some $j$ such that $d_{j+1}\sqsubseteq d_j$.
In this case the call either terminates directly or
recursively calls ${\sf solve}(b',i)$ for $b'={\bf true}$. Therefore, solving terminates also in this case.

It remains to prove that upon termination, a sound variable assignment $\sigma$ is found.
For $j=1,\dotsc,n$, and a variable assignment $\rho:\{y_1,\dotsc,y_n\}\to\D$, we consider
the system ${\cal E}_{\rho,j}$ defined by:
\[
y_i = f^\sharp_{\rho,i}\qquad(i=1,\dotsc,j)
\]
with $f^\sharp_{\rho,i}\,\sigma^\sharp = f_i(\rho\oplus\sigma^\sharp)$ for $\sigma^\sharp:\{y_1,\dotsc,y_j\}\to\D$.
Here, the operator $\oplus$ is meant to overwrite the values of $\rho$ by the corresponding values
of $\sigma^\sharp$ whenever $\sigma^\sharp$ is defined. Let $\underline{\cal E}_{\rho,j}$ denote the lower monotonization
of ${\cal E}_{\rho,j}$. We claim:
\begin{enumerate}
\item	Assume $\sigma_1$ is a post-solution of the system $\underline {\cal E}_{\rho,j}$.
	Then ${\sf solve}({\bf true},j)$ when started with $\rho_1=\rho\oplus\sigma_1$,
	returns with a variable assignment $\rho_2=\rho\oplus\sigma_2$ where
	$\sigma_2$ is still a post-solution of $\underline{\cal E}_{\rho,j}$.
\item	${\sf solve}({\bf false},j)$ returns a post-solution of $\underline{\cal E}_{\rho,j}$.
\end{enumerate}
We proceed by induction on $j$. For $j=0$, nothing must be proven.
Therefore assume $j>0$. Consider the first claim. As post-solutions of $\underline{\cal E}_{\rho,j}$
are preserved by each update which combines an old value $\sigma(y_i)$ with the value of the corresponding
right-hand side $f^\sharp_{\rho,i}$ for $\sigma$ by means of $\sqcap$ and thus also by $\narrow$, the claim follows.

For a proof of the second claim,
let us consider the sub-sequence of tail-recursive calls ${\sf solve}(b,j)$
where $b'$ remains ${\bf false}$.
Eventually this sequence ends with a last call where $b'$ is set to ${\bf true}$.
Let $\rho'$ denote the variable assignment before this update occurs.
Then $\rho'(y_j)\sqsupseteq f_j\,\rho'\sqsupseteq\underline f_j\,\rho'$.
Likewise, by induction hypothesis, $\rho'|_{\{y_{1},\ldots,y_{j-1}\}}$ is a post-solution
of $\underline{\cal E}_{\rho',j-1}$.
Altogether therefore, $\rho' = \rho\oplus\sigma'$ for some variable assignment $\sigma':\{y_1,\ldots,y_j\}\to\D$
which is a post-solution of $\underline{\cal E}_{\rho',j}$.
Accordingly, ${\sf solve}({\bf false},j)$ either directly terminates with $\rho'$, and the second claim follows,
or ${\sf solve}({\bf true},j)$ is called, and the second claim follows from the first one.
This completes the proof of the two claims.
Since the second claim, instantiated with $j=n$, implies that the variable assignment returned by the
algorithm is a post-solution of the lower monotonization of the system, it is sound. And by Lemma~\ref{l:lower0}.3,
it then is also a post-solution
of the original abstract system whenever all right-hand sides are monotonic.
\qed
   \section{Proof of Theorem~\ref{t:term-two}}\label{s:two-term:proof}

Assume that only finitely many variables are encountered during the run of the algorithm, i.e.,
from some point neither ${\sf dom}_0$ nor ${\sf dom}_1$ receive new elements. 
Since ${\sf solve}_0(y)$ is called before the variable $y$ is added to ${\sf dom}_1$,
and ${\sf solve}_0(y)$ enforces that $y$ is included in ${\sf dom}_0$, we have that
${\sf dom}_1\subseteq{\sf dom}_0$ throughout the algorithm. 
Due to the initial call ${\sf solve}_1(y_0,0)$, $y_0$ is contained in $Y_1$ implying the first item in the list.

Variables $y$ are added into sets ${\sf infl}[z]$ only during the evaluation of a call to ${\sf eval}_i$ and after
an appropriate call to ${\sf solve}_i$ --- implying that $y$ is contained in ${\sf dom}_i$ whenever ${\sf eval}_i$ 
was called inside a call ${\sf do\_var}_i(y)$.
Accordingly, all variables added to the priority queue necessarily are contained in ${\sf dom}_0$.
Thus, all variables for which ${\sf do\_var}_0$ is called at a call of ${\sf iterate}_0(n)$
are all contained in ${\sf dom}_0$, 
while all variables for which ${\sf do\_var}_1$ is called at a call of ${\sf iterate}_1(n)$ are already 
contained in ${\sf dom}_1$.
Therefore, we define $Y_i = {\sf dom}_i$ when the iteration has terminated for $i=0,1$.
We claim that for every priority $n$, the following holds:
\begin{enumerate}
\item	Every call ${\sf iterate}_0(n)$
during the evaluation of ${\sf solve}_0(0,y_0)$ terminates.
\item	Every call ${\sf iterate}_1(n)$
during the evaluation of ${\sf solve}_1(0,y_0)$
terminates as well.
\end{enumerate}

In order to prove the first claim, assume for a contradiction that there is some $n$ such that
the call ${\sf iterate}_0(n)$ does not terminate. 
Since $Y_0$ is finite, there must by a variable $y$
of maximal priority ${\sf prio}(y) \leq n$ so that ${\sf do\_var}_0(y)$ is evaluated infinitely often.
This means that from some point on, $y$ is the variable of maximal priority for which ${\sf do\_var}_0$ is called.
Let $d_i,i\geq 0$ denote the sequence of the new values for $y$. We claim that for every $i\geq 0$,
$d_{i+1} = d_i\widen a_i$ holds for some suitable value $a_i$. This holds if $y\in{\sf point}$ from the first
evaluation onward. Clearly, if this were the case, we arrive at a contradiction, as any such widening sequence
is ultimately stable.
Accordingly, it remains to prove that from the first evaluation onward, $y$ is contained in ${\sf point}$ ---
whenever ${\sf do\_var}_0(y)$ is called.
Assume for a contradiction that there is a first such call where $y$ is not contained in ${\sf point}$.
Assume that this call provided the $i$th value $d_i$ for $y$.
This means that, since the last evaluation of $f^\sharp_y$, no query to the value of $y$ during
the evaluation of lower priority variables has occurred. Accordingly, the set ${\sf infl}[y]$ does not contain
any lower priority variables, which means that no further variable is evaluated before the next call
${\sf do\_var}_0(y)$. But then this next evaluation of $f^\sharp_y$ will
return the value $a$. Subsequently, the queue $Q$ does no longer contain variables of priority less then or equal to $n$,
and therefore the iteration would terminate --- in contradiction to our assumption.

Now consider the second claim. For a contradiction now assume that there is some $n$ so that
the call  ${\sf iterate}_1(n)$ does not terminate. Since every call ${\sf iterate}_0(m)$
encountered during its evaluation is already known to terminate, we conclude that there must be a variable $y$
of priority less then or equal to $n$ so that ${\sf do\_var}_1(y)$ is evaluated infinitely often.
As before this means that from some point on, $y$ is the variable of maximal priority for which
${\sf do\_var}_1(y)$ is called.
Let $d_i,i\geq 0$ denote the sequence of the new values for $y$. We claim that for every $i\geq 0$,
$d_{i+1} = d_i\narrow a_i$ holds for some suitable value $a_i$. This holds if $d_i\sqsubseteq d_{i+1}$ and
$y\in{\sf point}$ from $i=1$ onward.
Again, if this were the case, we arrive at a contradiction, as any such widening sequence
is ultimately stable.
Accordingly, it remains to prove that from the first evaluation onward, $y$ is contained in ${\sf point}$ ---
whenever ${\sf do\_var}_1(y)$ is called. This, however, follows by the same argument as for
${\sf iterate}_0(y)$.
This completes the proof of the claim.

By the claim which we have just proven, each occurring call ${\sf iterate}_i(n)$ will terminate.
From that, the termination of the call ${\sf solve}_1(y_0,0)$ follows as stated by the theorem.
 
It remains to prove the remaining two assertions of the enumeration.
Again, we assume that only finitely many variables are encountered in a run of the local two-phase solver when started 
for a variable $y_0$, and assume that after some call to ${\sf do\_var}_i$, 
no further variable is added to ${\sf dom}_0$,
and likewise no further variable is added to ${\sf dom}_1$.
In order to prove the second assertion,
we prove that the following invariants hold before every call ${\sf do\_var}_i(y_1)$:
\begin{enumerate}
\item	For every variable $y$ in the current domain ${\sf dom}_0$, ${\sf infl}[y]$ contains (at least) all variables
	$z\not\in Q\cup\{y_1\}$ whose last evaluation of $f^\sharp_z$ has called ${\sf eval}_i(y)$;
\item	If $y\in{\sf dom}_0\backslash(Q\cup\{y_1\})$, then 
	$\sigma_0[y]\sqsupseteq f_y^\sharp(\underline\top\oplus\sigma_0)$;
\item	If $y\in{\sf dom}_1\backslash(Q\cup\{y_1\})$, then 
	$\sigma_1[y]\sqsupseteq\underline f_y^\sharp(\underline\top\oplus\sigma_1)$;
\end{enumerate}
Here, $\underline\top$ is the variable assignment which maps each variable in $Y_0$ to $\top$.
Here, we only prove the second invariant.
For that, consider a call to ${\sf do\_var}_1(y_1)$. 
If this is the very first call of ${\sf do\_var}_1$ for $y_1$,
then this occurs inside a call ${\sf solve}_1(y_1)$. Accordingly, the value 
$\sigma_1[y]$ has been initialized to $\sigma_0[y]$. Then we have, by the first invariant:
\[
\sigma_1[y_1] = \sigma_0[y_1]\sqsupseteq f^\sharp_{y_1}\,(\underline\top\oplus\sigma_0)\sqsupseteq
	\underline f^\sharp_{y_1}\,(\underline\top\oplus\sigma_0)
\]
At that moment, the priority queue
does not contain any variable $y$ with priority less or equal the priority of $y_1$,
implying that for all these $y$, $\sigma_1[y]\sqsupseteq\underline f_y^\sharp(\underline\top\oplus\sigma_1)$ holds.
This property is preserved by updating $\sigma_1[y_1]$ with a value exceeding 
$\underline f_{y_1}^\sharp(\underline\top\oplus\sigma_1)$.
Accordingly, all variables $z$ from ${\sf infl}[y_1]$ with priority less or equal to $y_1$ will subsequently
be iterated upon with ${\sf iterate}_1$. But since these variables $z$
satisfy $\sigma_1[z]\sqsupseteq\underline f_{z}^\sharp(\underline\bot\oplus\sigma_1)$, the invariant holds for the
calls of ${\sf do\_var}_1$ therein.

By construction, 0 is the maximal priority of any variable.
$y_0$ receives the least priority. Therefore, ${\sf solve}_1(0,y_0)$ returns with an empty queue $Q$.
By the second invariant the second assertion of the theorem follows.
\qed
   \section{Proof of Theorem~\ref{t:sslr-term}}\label{s:slr-term:proof}

By Theorem~\ref{t:local} the only condition  for \TSL\ to terminate is that only finitely many variables are
encountered. No further assumptions, e.g., w.r.t.\ monotonicity of right-hand sides must be made as in
\cite{Apinis13,Amato16}.
Upon termination, the algorithm is guaranteed to return sound results. The returned variable assignment is a (partial)
post-solution of the lower monotonization of the system, which means it may not necessarily be a post-solution of the
original system --- given that some right-hand sides are not monotonic.

Assume that only variables from the finite set $Y_0$ are encountered during the run of the algorithm.
We claim that for every priority $i$, the following holds:
\begin{enumerate}
\item	Every call ${\sf iterate}({\bf true}, i)$
during the evaluation of ${\sf solve}\,y_0$ terminates.
\item	Every call ${\sf iterate}({\bf false}, i)$
during the evaluation of ${\sf solve}\,y_0$
terminates as well.
\end{enumerate}

In order to prove the first claim, assume for a contradiction that there is some $i$ such that
the call ${\sf iterate}({\bf true}, i)$ does not terminate. Note that then any subsequent call to ${\sf do\_var}$ as well
as ${\sf iterate}$ will always be evaluated for the Boolean value ${\bf true}$.
Since $Y_0$ is finite, there is a variable $y$
of maximal priority ${\sf prio}(y) \leq i$ so that ${\sf do\_var}({\bf true},y)$ is evaluated infinitely often.
This means that from some point on, $y$ is the variable of maximal priority for which ${\sf do\_var}$ is called.
Let $d_i,i\geq 0$ denote the sequence of the new values for $y$. We claim that for every $i\geq 0$,
$d_{i+1} = d_i\narrow a_i$ holds for some suitable value $a_i$. This holds if $y\in{\sf point}$ from the first
evaluation onward. Clearly, if this were the case, we arrive at a contradiction, as any such narrowing sequence
is ultimately stable.
Accordingly, it remains to prove that from the first evaluation onward, $y$ is contained in ${\sf point}$ ---
whenever ${\sf do\_var}({\bf true},y)$ is called.
Assume for a contradiction that there is a first such call where $y$ is not contained in ${\sf point}$.
Assume that this call provided the $i$th value $d_i$ for $y$.
This means in particular that, since the last evaluation of $f^\sharp_y$, no query to the value of $y$ during
the evaluation of lower priority variables has occurred. Accordingly, the set ${\sf infl}[y]$ does not contain
any lower priority variables, which means that no further variable is evaluated before the next call
${\sf do\_var}({\bf true},y)$. But then this next evaluation of $f^\sharp_y$ will
return the value $a$. Subsequently, the queue $Q$ does no longer contain variables of priority $\leq i$,
and therefore the iteration would terminate --- in contradiction to our assumption.

Let us therefore now consider the second claim. For a contradiction now assume that there is some $i$ so that
the call  ${\sf iterate}({\bf false}, i)$ does not terminate. Since every call ${\sf iterate}({\bf true},j)$
encountered during its evaluation is already known to terminate, we conclude that there must be a variable $y$
of maximal priority $\leq i$ so that ${\sf do\_var}({\bf false},y)$ is evaluated infinitely often.
As before this means that from some point on, $y$ is the variable of maximal priority for which
${\sf do\_var}({\bf false},y)$ is called.
Let $d_i,i\geq 0$ denote the sequence of the new values for $y$. We claim that for every $i\geq 0$,
$d_{i+1} = d_i\widen a_i$ holds for some suitable value $a_i$ where
$y\in{\sf point}$ from $i=1$ onward.
Again, if this were the case, we arrive at a contradiction, as any such widening sequence
is ultimately stable.
Accordingly, it remains to prove that from the first evaluation onward, $y$ is contained in ${\sf point}$ ---
whenever ${\sf do\_var}({\bf false},y)$ is called. This, however, follows by the same argument as for
${\sf iterate}({\bf true},y)$.
This completes the proof of the claim.

Now assume that only finitely many variables are encountered in a run of \TSL\ when started for a variable $x$,
and assume that after some call to ${\sf do\_var}$, no further variable is encountered.
Let $Y_0$ denote this set of variables.
By the claim which we have just proven, each subsequent call to the function ${\sf iterate}$ will terminate.
From that, the termination of the call ${\sf solve}\,y_0$ follows as stated by the theorem.
 
It remains to prove the second assertion.
We remark that,	whenever a new variable is encountered, it is added into the set ${\sf dom}$ and never removed.
Let $Y_0$ again denote the finite set of variables encountered during
${\sf solve} \,y_0$, i.e., the final value of ${\sf dom}$.
In particular, $y_0$ is contained in $Y_0$. In order to prove the second assertion,
we note that the following invariants hold before every call ${\sf do\_var}(b,y_1)$:
\begin{enumerate}
\item	For every variable $y$ in the current domain ${\sf dom}$, ${\sf infl}[y]$ contains (at least) all variables
	$z\not\in Q\cup\{y_1\}$ whose last evaluation of $f^\sharp_z$ has called ${\sf eval}\,y$;
\item	If $y\in{\sf dom}\backslash(Q\cup\{y_1\})$, then $\sigma[y]\sqsupseteq\underline f_y^\sharp(\underline\top\oplus\sigma)$;
\item	If $b={\bf true}$, then $\sigma[y_1]\sqsupseteq\underline f_{y_1}^\sharp(\underline\top\oplus\sigma)$.
\end{enumerate}
Here, $\underline\top$ is the variable assignment which maps each variable in $Y_0$ to $\top$.
In order to see the second statement, we observe that iteration on a variable always starts with
the flag ${\bf false}$. Now consider a call to ${\sf do\_var}({\bf false},y_1)$ where $b'$ is set
to ${\bf true}$. This is the case when $\sigma[y_1]\sqsupseteq{\it tmp}$  where ${\it tmp}$ is the
value of the last evaluation of the right-hand side of $y_1$.
Accordingly,
\[
\sigma[y_1]\sqsupseteq f_{y_1}^\sharp\,(\underline\top\oplus\sigma)\sqsupseteq
	\underline f_{y_1}^\sharp\,(\underline\top\oplus\sigma)
\]
At that moment, the priority queue
does not contain any variable $y$ with priority less or equal the priority of $y_1$,
implying that for $y$, $\sigma[y]\sqsupseteq\underline f_y^\sharp(\underline\top\oplus\sigma)$
holds.
This property is preserved by updating $\sigma[y_1]$ with a value exceeding $\underline f_y^\sharp(\underline\top\oplus\sigma)$.
Accordingly, all variables $z$ from ${\sf infl}[y_1]$ with priority less or equal to $y_1$ will subsequently
be iterated upon with $b={\bf true}$. But since these
satisfy $\sigma[z]\sqsupseteq\underline f_{z}^\sharp(\underline\top\oplus\sigma)$, the invariant holds for the
calls of ${\sf do\_var}$ therein.

By construction,
$y_0$ receives the least priority. Therefore, ${\sf solve}\,y_0$ returns with an empty queue $Q$.
By the second invariant the second assertion of the theorem follows.
\qed

 }
\end{document}